\documentclass{llncs}

\usepackage{amsfonts}
\usepackage{graphicx,color,subfig}
\usepackage{bm}

\usepackage{multirow}

\usepackage[utf8x]{inputenc}
\usepackage[T1]{fontenc}
\usepackage{lmodern} 
\usepackage{url}
\usepackage{amssymb,amsmath,amsthm,amsfonts}
\usepackage{fixltx2e}
\usepackage{hyperref}
\usepackage{svg}
\usepackage{graphicx}
\usepackage{fancyhdr} 
\usepackage{makeidx}
\usepackage{placeins}
\usepackage{verbatim}
\usepackage{xfrac}
\usepackage[linesnumbered,ruled,vlined]{algorithm2e}
\usepackage{float}
\DontPrintSemicolon
\usepackage{amssymb}
\usepackage{epstopdf}
\usepackage{multicol}
\usepackage{epigraph}
\usepackage{float}
\usepackage[inline]{enumitem}
\usepackage{mathtools}
\usepackage{enumitem}
\usepackage{nameref}
\DeclarePairedDelimiter\ceil{\lceil}{\rceil}
\DeclarePairedDelimiter\floor{\lfloor}{\rfloor}

\theoremstyle{plain}

\usepackage{mathtools} 

\theoremstyle{Proposition}
\newtheorem{prop}[theorem]{Proposition} 

\theoremstyle{definition}

\begin{document}

\title{Coalition Resilient Outcomes in Max $k$-Cut Games}

\author{%
Raffaello Carosi\inst{1}  \and
Simone Fioravanti\inst{2} \and
Luciano Gual\`{a}\inst{2} \and 
Gianpiero Monaco\inst{3} 
}%

\institute{
Gran Sasso Science Institute, Italy,\\
\email{raffaello.carosi@gssi.it}\\ 
\and
University of Rome ``Tor Vergata'', Italy,\\
\email{simonefi92@gmail.com}\\
\email{guala@mat.uniroma2.it}
\and
DISIM - University of L'Aquila, Italy,\\
\email{gianpiero.monaco@univaq.it}}

\maketitle

\begin{abstract}
We investigate strong Nash equilibria in the \emph{max $k$-cut game}, where we are given an undirected edge-weighted graph together with a set $\{1,\ldots, k\}$ of $k$ colors. Nodes represent
players and edges capture their mutual interests.
The strategy set of each player $v$ consists of the $k$ colors. When players select a color they induce a $k$-coloring or simply a coloring. Given a coloring, the \emph{utility} (or \emph{payoff}) of a player $u$ is the sum of the weights of the edges $\{u,v\}$ incident to $u$, such that the color chosen by $u$ is different from the one chosen by $v$.
Such games form some of the basic payoff
structures in game theory, model lots of real-world scenarios with selfish agents and extend or are related to several fundamental classes of games. 

Very little is known about the existence of
strong equilibria in max $k$-cut games. In this paper we make some steps forward in the comprehension of it. We first show that improving deviations performed by minimal coalitions can cycle, and thus answering negatively the open problem proposed in \cite{DBLP:conf/tamc/GourvesM10}. Next, we turn our attention to unweighted graphs. We first show that any optimal coloring is a 5-SE in this case. 
Then, we introduce $x$-local strong equilibria, namely colorings that are resilient to deviations by coalitions such that the maximum distance between every pair of nodes in the coalition is at most $x$. We prove that $1$-local strong equilibria always exist. Finally, we show the existence of strong Nash equilibria in several interesting specific scenarios.
\end{abstract}

\section{Introduction}
We consider the \emph{max $k$-cut game}. This is played on an undirected edge-weighted graph where the $n$ nodes correspond to the players and the edges capture their mutual interests. The strategy space of each player is a set $\{1,\ldots, k\}$ of $k$ available colors (we assume that the colors are the same for each player). When players select a color they induce a $k$-coloring or simply a coloring. Given a coloring, the \emph{utility} (or \emph{payoff}) of a player $u$ is the sum of the weights of edges $\{u,v\}$ incident to $u$, such that the color chosen by $u$ is different from the one chosen by $v$. The objective of every player is to maximize its own utility. 
 
This class of games forms some of the basic payoff structures in game theory, and can model lots of real-life scenarios. Consider, for example, a set of companies that have to decide which product to produce in order to maximize their revenue. Each company has its own competitors (for example the ones that are in the same region), and it is reasonable to assume that each company wants to minimize the number of competitors that produce the same product. 
Another possible scenario is in a radio setting; radio towers are players and their goal is selecting a frequency such that neighboring radio-towers have a different one in order to minimize the interference.

In such games on graphs it is beneficial for each player to anti-coordinate its choices with the ones of its neighbors (i.e., selecting a different color). As a consequence, the  players may attempt to increase their utility by coordinating their choices in groups (also called coalitions). Therefore, in our studies we focus on equilibrium concepts that are resilient to deviations of groups. Along this direction, a very classic notion of equilibrium is the strong Nash equilibrium (SE) \cite{A60} that is a coloring in which no coalition, taking the actions of its complements as given, can cooperatively deviate in a way that benefits all of its members, in the sense that every player of the coalition strictly improves its utility. The notion of SE is a very strong equilibrium concept. A weaker one is the notion of \emph{$q$-Strong Equilibrium} ($q$-SE), for some $q \leq n$, where only coalitions of at most $q$ players are allowed to cooperatively change their strategies. Notice that the $1$-SE is equivalent to the Nash equilibrium (NE), while the $n$-SE is equivalent to the SE.

When it exists, an SE is a very robust state of the game and it is also more sustainable than an NE. However, while NE always exists in these games \cite{DBLP:conf/cocoon/CarosiM18,DBLP:phd/de/Hoefer2007,DBLP:conf/sagt/KunPR13}, little is known about the existence of strong equilibria in Max $k$-cut games. Indeed, to the best of our knowledge, there are basically two papers of the literature dealing with such issue. In \cite{DBLP:conf/wine/GourvesM09} the authors show that an optimal strategy profile (or optimal coloring), i.e., a coloring that maximizes the sum of the players' utilities or equivalently, a coloring that maximizes the $k$-cut, is an SE for the max $2$-cut game, and it is a $3$-SE, for the max $k$-cut game, for any $k \geq 2$.  Moreover, they further show that an optimal strategy profile is not necessarily a $4$-SE, for any $k \geq 3$. In \cite{DBLP:conf/tamc/GourvesM10} they show that, if the number of colors is at least the number of players minus two, then an optimal strategy profile is an SE. Finally, they show that the dynamics, where at each step a coalition can deviate so that all of its members strictly improve their utility by changing strategy, can cycle. The main consequence of this latter fact is that no \emph{strong potential function}\footnote{See Section \ref{sec:preliminaries} for the definition of strong potential function.} can exist for the game, and hence the existence of an SE cannot be proved by simply exhibiting it. It is worth noticing that strong potential functions are one of the main tools used to prove the existence of an SE.  

All the above results suggest that it is hard to understand whether SE always exist for max $k$-cut games. In this paper, although we do not prove or disprove that every instance
of the max $k$-cut game possesses a strong equilibrium, we make some step forward in the comprehension of it.

\smallskip
\noindent{\bf Our results.} As pointed out in \cite{DBLP:conf/tamc/GourvesM10}, sometimes the existence of an SE is proved by means of a potential function in which the set of deviating coalitions is restricted to minimal coalitions only, where a coalition is minimal if none of its proper subsets can perform an improvement themselves (see for example \cite{HarksKM09}). Understanding whether this approach can be used in the max $k$-cut game is mentioned as an open problem in \cite{DBLP:conf/tamc/GourvesM10}. We answer this question negatively (see Proposition \ref{prop:cyclepotential}) by showing an instance in which there is a cycle of improving deviations performed by minimal coalitions only.

We then focus on the unweighted case, where the utility of a player in a coloring is simply the number of neighbors with different color from its own, and we provide some non-trivial existential results for it. In particular, in Section \ref{sec:existence_5-SE} we show that $5$-SE always exist for the max $k$-cut game. This is an improvement with respect to the existence of $3$-SE \cite{DBLP:conf/wine/GourvesM09}. 

Besides $q$-SE, we also consider another equilibrium concept that is weaker than the notion of SE. Observe that in a $q$-SE two players can form a coalition even if they are far from each other in the graph. This is unrealistic in many practical scenarios. In oder to encompass this aspect, in Section \ref{sec:local_SE}, we introduce the concept of \emph{$x$-Local} SE ($x$-LSE). A coloring is an $x$-LSE if it is resilient to deviations by coalitions such that the shortest path between every pair of nodes in the coalition is long at most $x$. Therefore, 
the notion of $x$-LSE also takes into account that certain players may not have the possibility of communicating to each other and thus to form a coalition. This seems an important point to consider when modeling a situation of strategic interaction between agents. Here we suppose that the input graph also represents knowledge between players, that is two nodes know each other if they are connected by an edge. In this paper we focus on the case $x=1$, that is each player in the coalition must have a social connection (namely an edge) towards every deviating player. We show that, for any $k$, a $1$-LSE always exists. Interestingly enough, our analysis also provides a characterization of the set of local strong equilibria which relates $1$-LSE to $q$-SE.

Finally, in Section \ref{sec:special cases}, we show that an SE always exists for some special classes of unweighted graphs. More precisely, in Corollary \ref{cor:girth}, we prove that in graphs with large girth, any optimal strategy profile is an SE, for any $k \geq 2$. Moreover, in Proposition \ref{prop:degree}, we prove that whenever the number of colors $k$ is large enough with respect to the maximum degree of the graph, then any optimal strategy profile is an SE.     
  
Some proofs have been moved to the appendix.

\smallskip
\noindent{\bf Further related work.} The max $k$-cut game has been first investigated in \cite{DBLP:phd/de/Hoefer2007,DBLP:conf/sagt/KunPR13}, where the authors show that, when the graph is unweighted and undirected, it is possible to compute a Nash Equilibrium in polynomial time by exploiting the potential function method. When the graph is weighted undirected, even if the potential function ensures the existence of NE the problem of computing an equilibrium is PLS-complete even for $k=2$~\cite{SY91}. In fact, for such a value of $k$, it coincides with the classical max cut game. In \cite{DBLP:conf/cocoon/CarosiM18} the authors show the existence of NE in generalized max $k$-cut games where players also have an extra profit depending on the chosen color. When the graph is directed, the max $k$-cut game in general does not admit a potential function. Indeed, in this case, even the problem of understanding whether they admit a Nash equilibrium is NP-complete for any fixed $k \geq 2$ \cite{DBLP:conf/sagt/KunPR13}. In \cite{DBLP:conf/atal/CarosiFM17}, the authors present a randomized polynomial time algorithm that computes a constant approximate Nash equilibrium for a large class of directed unweighted graphs.

Studies on the performance of Nash equilibria  and strong Nash equilibria can be found in \cite{DBLP:conf/cocoon/CarosiM18,DBLP:phd/de/Hoefer2007,DBLP:conf/sagt/KunPR13} and in \cite{Feldman2015,DBLP:conf/wine/GourvesM09,DBLP:conf/tamc/GourvesM10}, respectively.

A related stream of research considers coordination games. The idea is that agents are rewarded for choosing common strategies in order to capture the influences. Apt et al.~\cite{DBLP:journals/ijgt/AptKRSS17} propose a coordination game modeled as an undirected graph where nodes are players and each player has a list of allowed colors. Given a coloring, an agent has a payoff equal to the number of adjacent nodes with its same color. The authors show that NE and $2$-SE always exists, and give an example in which no $3$-SE exists. Moreover, they prove that strong equilibria exist for various special cases. 

Panagopoulou and Spirakis~\cite{anti2} study games where Nash equilibria are proper node coloring in undirected unweighted graphs setting. In particular, they consider the game where each agent $v$ has to choose a color among $k$ available ones and its payoff is equal to the number of nodes in the graph that have chosen its same color, unless some neighbor of $v$ has chosen the same color, and in this case the payoff of $v$ is $0$. They prove that this is a potential game and that a Nash equilibrium can be found in polynomial time. 

Max $k$-cut games are related to many other fundamental games considered in the scientific literature. One example is given by the graphical games introduced in \cite{KLS01}. In these games the payoff of each agent depends only on the strategies of its neighbors in a given social knowledge graph defined over the set of the agents, where an arc $(i,j)$ means that $j$ influences $i$'s payoff. Max $k$-cut games can also be seen as a particular hedonic game (see \cite{AS14} for a nice introduction to hedonic games) with an upper bound (i.e., $k$) to the number of coalitions. Specifically, given a $k$-coloring, the agents with the same color can be seen as members of the same coalition of the hedonic game. In order to get the equivalence among the two games, the hedonic utility of an agent $v$ can be defined as the overall number of its neighbors minus the number of agents of its neighborhood that are in the same coalition. Nash equilibria issues in hedonic games have been largely investigated under several different assumptions \cite{BFFMM18,BJ02,MMV18} (just to cite a few).     

Concerning local coalitions, a notion of equilibrium close in spirit to our LSE has been studied in the context of network design games in \cite{DBLP:conf/podc/LeonardiS07}. Moreover, locality aspects have been also considered when restricting the strategy space in single-player deviations (see for example \cite{DBLP:conf/spaa/BiloGLP14,Cord-LandwehrL15}).

Finally, it is worth mentioning the classical optimization max cut problem, a very famous problem in graph theory that was proven to be NP-Hard by Karp~\cite{np}. 

\section{Preliminaries}\label{sec:preliminaries}

Let $G = (V,E,w)$ be an undirected weighted graph, where $\lvert V\rvert = n$, $\lvert E\rvert = m$, and $w: E \rightarrow \mathbb{R}_{+}$. Let $\delta^v(G) = \sum_{u\in V: \{v,u\}\in E}w(\{v,u\})$ denote the \emph{degree} of $v$, that is the sum of the weights of all the edges incident to $v$. Let $\delta^{M}(G) = max_{v\in V}\delta^v(G)$ denotes the maximum degree in $G$. Given a set of nodes $V'\subseteq V$, let $G(V') = (V', E', w)$ be the subgraph induced by $V'$, where $E' = \{\{v,u\}\in E \mid v\in V' \wedge u\in V'\}$.
For any pair of nodes $v,u \in V$, the \emph{distance} $dist_G(v,u)$ between $v$ and $u$ in $G$ is equal to the length of the shortest path from $v$ to $u$\footnote{Even if the graph is weighted, we consider here the \emph{hop-distance}, where the length of a path is defined as the number of its edges.}. 

Given $G$ and a set of colors $K =\{1,\ldots k\}$, the \emph{max k-cut} problem is to partition the vertices into $k$ subsets $V_1, \ldots, V_k$ such that the sum of the weights of the edges having the endpoints in different sets is maximized. A strategic version of the max k-cut problem is the \emph{max k-cut game}, and it is defined as follows. There are $\lvert V\rvert$ players, and each node of $G$ is controlled by exactly one rational player. Players have the same strategy set, and it is equal to the set of colors $\{1,\ldots k\}$. A \emph{strategy profile}, or \emph{coloring} $\sigma : V \rightarrow K$, is a labeling of nodes of $G$ in which each player $v$ is colored $\sigma(v)$. Given a coloring $\sigma$, let $E(\sigma) = \{\{u,v\}: \sigma(u) \neq \sigma(v)\}$ be the edges that are proper with respect to $\sigma$, and let $\delta_u^{i}(\sigma) = \sum_{v\in V}w(\{u,v\})_{\sigma(v) = i}$ be the sum of the weights of the edges incident to $u$ and towards nodes colored $i$ in $\sigma$. The utility (or payoff) of player $u$ is defined as $\mu_u(\sigma) = \sum_{v\in V: \{u,v\}\in E \wedge \sigma(u)\neq \sigma(v)} w(\{u,v\})$.  
The \emph{cut-value}, or \emph{size of the cut}, of a coloring $S(\sigma)$ is defined as follows: $S(\sigma)=\sum_{\{u,v\}\in E \wedge \sigma(u)\neq \sigma(v)} w(\{u,v\})$.  
The \emph{social welfare} of a coloring $\sigma$ is defined as the sum of players' utilities, that is $SW(\sigma) = \sum_{v\in V}\mu_v(\sigma) = 2 S(\sigma)$. Moreover, an optimal strategy profile (or optimal coloring) is defined to be a strategy profile which maximizes the sum of the players' utilities and thus the cut-value.

Given a coalition $C \subseteq V$ and a coloring $\sigma$, let $C_K(\sigma) = \{i \in K \mid \exists\text{ }v\in C \text{ s.t. } \sigma(v) = i \}$ be the set of colors used by the coalition $C$ in $\sigma$. Moreover, for each color $i$, let $C_i(\sigma) = \{v\in C \mid \sigma(v) = i\}$ be the set of players in $C$ that are colored $i$ in $\sigma$. 

Given a strategy profile $\sigma$, a player $v$ and a coalition $C$, we denote by $\sigma_{-v}$ and $\sigma_{-C}$ the strategy profile $\sigma$ besides the strategy played by $v$ and by $C$, respectively. Moreover, we denote by $\sigma_C$ the coloring $\sigma$ restricted only to players in $C$, and we use $(\sigma_{-v}, \sigma(v))$ and $(\sigma_{-C}, \sigma_C)$ to denote $\sigma$.

A profile $\sigma$ is a \emph{Nash Equilibrium} (NE) if no player can improve its payoff by deviating unilaterally from $\sigma$, that is, $\mu_v(\sigma_{-v}, i) \leq \mu_v(\sigma)$ for each player $v \in V$ and for each color $i \in K$. For each $1 \leq q \leq n$, $\sigma$ is a \emph{$q$-Strong Equilibrium} ($q$-SE) if there exists no coalition $C$ with $\lvert C\rvert \leq q$ that can cooperatively deviate from $\sigma_{C}$ to $\sigma_C'$ in such a way that every player in $C$ strictly improves its utility in $(\sigma_{-C},\sigma_{C}')$. The $1$-strong equilibrium is equivalent to the Nash equilibrium, while for $q=n$ an $n$-strong equilibrium is called \emph{strong equilibrium (SE)}. When a coalition $C$ deviates so that all of its members strictly improve their utility, then we say it performs a \emph{strong improvement}. A strong improvement is said to be \emph{minimal} if no proper subsets of the deviating coalition can perform an improvement themselves, and the coalition itself is said to be minimal. A strong improving dynamics (shortly dynamics) is a sequence of strong improving moves. A game is said to be convergent if, given any initial state, any sequence of improving moves leads to a strong Nash equilibrium. Given a coloring $\sigma$, if a coalition $C$ induces a new coloring $\sigma'$ after deviating, then we say that the set of edges $E(\sigma')\backslash E(\sigma)$ \emph{enters} the cut, and that the set of edges $E(\sigma)\backslash E(\sigma')$ \emph{leaves} the cut.

A \emph{potential function} $\Phi$ is a function mapping strategy profiles into real values in such a way that, for each coloring $\sigma$ and each player $v$, whenever $v$ can profitably deviate from $\sigma$ yielding a new coloring $\sigma'$, it holds that $\Phi(\sigma') > \Phi(\sigma)$. When this is true also for profitably deviations performed by coalitions, the function is called \emph{strong potential function}.

We conclude this section by stating some properties about minimal coalitions that will be useful later. The proofs of these properties can be found in the appendix.

\begin{prop}\label{prop:propminimalcoalition}
Let $\sigma$ be a coloring, and let $C$ be a minimal coalition that can perform a strong improvement from $\sigma$. Let $\sigma'$ be the resulting coloring. Then, the following properties hold: (i) $C_K(\sigma) = C_K(\sigma')$; and (ii) if $G(C)$ is acyclic, then changing from $\sigma$ to $\sigma'$ strictly increases the size of the cut.
\end{prop}

\section{Non-existence of a minimal strong potential function}\label{sec:looppotential}

\begin{figure}[h]
\begin{center}
\includegraphics[scale=0.45]{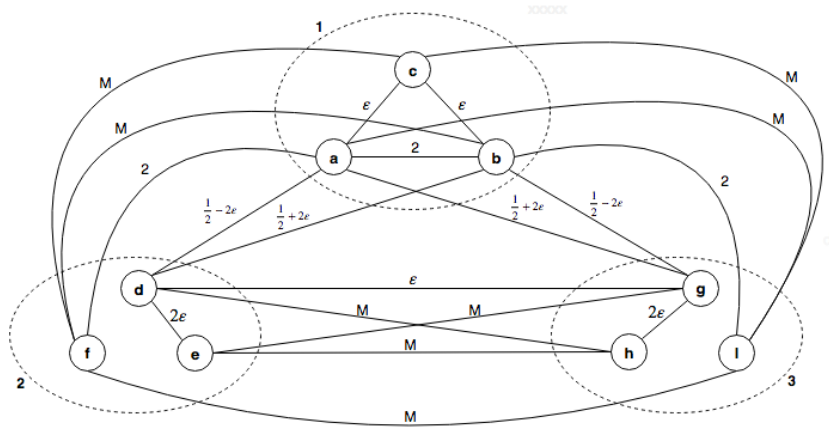}
\caption{Instance for which the strong improvement dynamics cycles.}\label{fig:minimalCycle}
\end{center}
\end{figure}

In this section we focus on weighted graphs and, as discussed in the introduction, we close an open problem stated by Gourv{\`{e}}s and Monnot in \cite{DBLP:conf/tamc/GourvesM10} by providing an instance in which there is a cycle of improving deviations performed by minimal coalitions only. More specifically, the loop is composed by the deviation of a clique, followed by four improvements performed by single players.

\begin{prop}\label{prop:cyclepotential}
No strong potential function exists for the max k-cut game, even if only minimal coalitions are allowed to deviate.
\end{prop}
\begin{proof}
Consider the graph $G$ and the coloring $\sigma$ depicted in Figure \ref{fig:minimalCycle}, where, if a node $v$ is contained in the dashed ellipse labeled $i$, then $v$ is colored $i$ in $\sigma$, and where $M$ and $\varepsilon$ denote a very large and small positive value, respectively. 

Consider coalition $C = \{a,b,d,g\}$ and consider the deviation $\sigma'$ where $\sigma'(a) = 2$, $\sigma'(b) = 3$, $\sigma'(d) = 1$, $\sigma'(g) = 1$.
It is easy to check that this deviation is profitable for players in $C$. In fact, they all improve their utility by $\varepsilon$.

Player $a$ has utility $3 + M$ in $\sigma$, and since it has an edge of weight $M$ towards node $i$ having color $3$, $a$ can only deviate to color $2$. Hence,
$a$ strictly improves its utility from $3 + M$ to $3 + M+ \varepsilon$ only if node $d$ changes color. Analogously, in $\sigma$, $d$ has one edge of weight $M$ towards node $h$ having color $3$. Thus, $d$ can only switch to color $1$ and this is convenient for it only if both $a$ and $b$ leave color $1$. If this happens, $d$'s utility increases by at least $\varepsilon$. Similarly to $a$, player $b$ deviates to color $3$ only if player $g$ switches to color $1$, and this happens only if both $a$ and $b$ deviates too. To sum up, both $d$ and $g$ deviate if and only if $a$ and $b$ deviate too. Thus, $C$ is minimal. Note that edge $\{d,g\}$ becomes monochromatic, but $d$ and $g$'s new payoffs make the deviation worth it anyway, since they both increase their utility by $\varepsilon$.

After the players in $C$ jointly deviate from $\sigma$, player $a$, who is now colored $2$, can go back to color $1$, improving its utility from $3 + \varepsilon + M$ to $4 + M$. Because of $a$'s deviation, $d$'s utility goes down to $1/2 +4\varepsilon + M$. Thus, it goes back to color $2$, achieving $1 + \varepsilon + M$. Also $g$, whose utility is now $1/2 + \varepsilon + M$, deviates to its old color in $\sigma$, that is color $3$, and it gets $1/2 + 3\varepsilon + M$. In this configuration $b$' utility is 
$5/2 +3\varepsilon +M$. Thus going back to color $1$ its utility improves to $3 + M$ and we are now back to the initial configuration $\sigma$. 
\end{proof}


\section{The existence of a 5-SE in unweighted graphs}\label{sec:existence_5-SE}

From now on we will focus on unweighted graphs. In \cite{DBLP:conf/wine/GourvesM09} it is shown that in the weighted case, any optimal strategy profile is always a $3$-SE, and there are weighted graphs in which every optimal coloring is not a $4$-SE. 
In this section we improve this result for unweighted graphs, by showing that a $5$-SE always exists. This also establishes a separation between the weighted and unweighted case. In particular, we show that by performing minimal strong improvements with coalitions of size at most five, the cut value increases. It implies that the cut value is a potential function and thus the dynamics converges to $5$-SE. We start by showing a simple lemma that is used in the rest of the section.

\begin{lemma}\label{lemma:2}
Let $\sigma$ be an NE and let $C$ be a minimal coalition which would profit by deviating from $\sigma$ to $\sigma'$.
If there exists two players $u,x\in C$ such that:
\begin{enumerate}
\item[(i)] $\sigma(u) \neq \sigma(x)$
\item[(ii)] $\sigma'(u) = \sigma(x)$
\item[(iii)] $\{y\in C| \{u,y\}\in E, \sigma(y) = \sigma(x)\} = \{x\}$ 
\end{enumerate}
then $\delta_u^{\sigma(u)}(\sigma) = \delta_u^{\sigma'(u)}(\sigma)$. Moreover, if there exists a third player $v\in C$ such that $\sigma(v) \neq \sigma(x)$ and $\sigma'(v) = \sigma(x)$, then $\{u,v\}\notin E$.
\end{lemma}

\begin{proof} Since $\sigma$ is an NE we know that $u$ cannot improve its utility by deviating alone to $\sigma'(u)$ which implies:
\[
\delta_u^{\sigma(u)}(\sigma) \leq \delta_u^{\sigma'(u)}(\sigma).
\]

By (iii) we know that, moving from $\sigma$ to $\sigma'$, the only neighbor of $u$ which leaves $\sigma'(u)$ is $x$ which means that its new neighbors colored $\sigma'(u)$ are at least $\delta^{\sigma'(u)}_u(\sigma) - 1$, because, a priori, other players could move to the same strategy in $\sigma'$, hence
$\delta^{\sigma'(u)}_u(\sigma') \geq \delta^{\sigma'(u)}_u(\sigma) - 1$. Moreover, player $u$ strictly improves its utility, which means $\delta^{\sigma'(u)}_u(\sigma')<\delta_u^{\sigma(u)}(\sigma)$. Using both inequalities we obtain: 
\[
\delta_u^{\sigma(u)}(\sigma) > \delta_u^{\sigma'(u)}(\sigma) -1.
\]
As a consequence, we have $\delta_u^{\sigma(u)}(\sigma) \ge \delta_u^{\sigma'(u)}(\sigma)$, and hence $\delta_u^{\sigma(u)}(\sigma) = \delta_u^{\sigma'(u)}(\sigma)$, which in turn implies  in particular that $\delta^{\sigma'(u)}_u(\sigma') = \delta^{\sigma'(u)}_u(\sigma) - 1$, i.e. $u$'s utility improves exactly by one. Thus, given a player $v$ like in the hypothesis, if $\{u,v\} \in E$ then $u$'s utility would not increase after the deviation.
\end{proof}

Proposition \ref{prop:propminimalcoalition} and Lemma \ref{lemma:2} can be used to prove the following proposition, which shows that when the size of a deviating coalition $C$ is related in a certain way to the number of colors used by the players in $C$, then the improving deviation always increases the size of the cut.

\begin{prop}\label{prop:propcoloring}
Let $\sigma$ be an NE and let $C$ be a minimal coalition which would profit by deviating from $\sigma$ to $\sigma'$.
If $|C_K(\sigma)| \in \{2, \lvert C\rvert - 1, \lvert C\rvert \}$, then the deviation strictly improves the size of the cut.
\end{prop}

Gourv{\`{e}}s and Monnot \cite{DBLP:conf/wine/GourvesM09} show that in weighted graphs an optimal solution is always a $3$-strong equilibrium, that is, it is resilient to any joint deviation by at most three players. Proposition \ref{prop:propcoloring} already extends this result since it implies that unweighted graphs admit a potential function when minimal coalitions of at most four players are allowed to deviate, implying that $4$-SE always exists. We now prove that the cut value is a potential function even when the deviation is extended to coalitions of size at most five. This implies that a $5$-SE always exists in unweighted graphs.

\begin{theorem}\label{prop:propcoalition}
Any optimal strategy profile is a 5-SE.
\end{theorem}

\section{Local strong equilibria}\label{sec:local_SE}

In this section we introduce and discuss local strong equilibria. As our main result, we show that, for any $k$, such an equilibrium always exists. Interestingly enough, our analysis also provides a characterization of the set of local strong equilibria which relates them to $q$-SE. 

Let $C\subseteq V$ be a set of players. We say that $C$ is an \emph{$x$-local coalition} if the distance in $G$ between any two players in $C$ is at most $x$. Moreover, we define an \emph{$x$-Local Strong Equilibrium} ($x$-LSE) to be a coloring in which no $x$-local coalition can profitably deviate. In this section, we will consider only the case $x = 1$, that is, the coalition $C$ induces a clique. We will use LSE in place of $1$-LSE.

Let us introduce some additional notation. Given a node $u$ and a strategy profile $\sigma$, we denote by $c_u(\sigma)$ the \emph{cost} of $u$ in $\sigma$, namely the number of neighbors of $u$ that have the same color of $u$ in $\sigma$, i.e. $\delta_u^{\sigma(u)}(\sigma)$. Notice that $c_u(\sigma)=\delta_u - \mu_u(\sigma)$. Given a coalition $C$, we also define $c_{u,C}(\sigma) = |\{(u,v)\in E |v \in C, \sigma(v)=\sigma(u)\}|$.

We now prove a technical lemma which gives some necessary conditions for a clique to deviate profitably from an NE.

\begin{lemma}\label{lm:tech1}
Let $\sigma$ be an NE. Suppose there exists a deviation $\sigma'$ such that all the members of $C$ can lower their cost changing from $\sigma$ to $\sigma'$. The following conditions must hold :
\begin{enumerate}
\item[(i)]$|C_i(\sigma)|=|C_i(\sigma')|$ for all $i=1,\ldots,k$;
\item[(ii)] $c_u(\sigma)-c_u(\sigma')=1 \;$ for all $u \in C$;
\item[(iii)] $c_u(\sigma)=c_u(\sigma_{-u},\sigma'(u))\;$ for all $u \in C$.
\end{enumerate}
\end{lemma}

The following lemma underlines a very interesting property of deviating cliques, which allows us to study only cliques formed by at most $k$ players.

\begin{lemma}\label{lm:tech2}
Let $C$ be a clique which profits by deviating from an NE $\sigma$ to $\sigma'$. Then there exists $j \leq |C_K(\sigma)|$ and a subcoalition $C'=\{u_i,\ldots,u_j\} \subseteq C$ whose players can improve their payoffs by deviating alone to the strategy they use in $\sigma'$.
Moreover it holds :
\begin{enumerate}
\item[(i)] $\sigma'(u_i)=\sigma(u_{i+1})$ for all $i=1,\ldots,j-1$;
\item[(ii)] $\sigma'(u_{j})=\sigma(u_1)$.
\end{enumerate}

\begin{proof}
Consider a player $v_1 \in C$ and assume without loss of generality that $\sigma(v_1)=1$ and $\sigma'(v_1)=2$.
By Lemma \ref{lm:tech1}, there is at least one player, say $v_2$, in $C$ such that $\sigma(v_2)=2$. If $\sigma'(v_2)=1$, then $C'=\{v_1,v_2\}$. Otherwise, call without loss of generality $\sigma'(v_2)=3$. Then, once again by Lemma \ref{lm:tech1} there is a node in $C$, say $v_3$, with $\sigma(v_2)=2$. We can iterate this argument until we get a node $v_h$ with $\ell :=\sigma'(v_h) \in \{1,2,\dots,h-1\}$. 
We set $C'=\{v_\ell,v_{\ell+1},\dots,v_{h} \}$ and set $j=|C'|$. Notice that $C'$ already satisfies the properties 1 and 2 of the statement of the lemma (observe also that it could be $C'=C$). 

Let $\sigma^*$ be the strategy profile in which the players in $C'$ play as in $\sigma'$ while the others play as in $\sigma$. It remains to show that all players in $C'$ is improving its utility by changing from $\sigma$ to $\sigma^*$.
Let $\nu=\sigma^*(v_i)$. By definition of $\sigma^*$, we claim that $c_{v_i}(\sigma^*) = c_{v_i}(\sigma_{-v_i}, \nu) -1$. This is true because there is exactly one player in $C'$ that leaves color $\nu$ and thus the number of $v_i$'s neighbors with such a color decreases by 1. Hence,
\begin{align*}
c_{v_i}(\sigma^*) & = c_{v_i}(\sigma_{-v_i},\nu)-1 \\
& = c_{v_i}(\sigma_{-v_i},\sigma'(v_i))-1 \\
& = c_{v_i}(\sigma)-1,
\end{align*}
\noindent where in the last equality we used property (iii) of Lemma \ref{lm:tech1} on $\sigma'$.
\end{proof}
\end{lemma}

Lemma \ref{lm:tech2} allows us to prove the main result of this section, which is the following:

\begin{theorem}
Any optimal strategy profile is an LSE.
\end{theorem}
\begin{proof}
Let $\sigma$ be an optimal strategy profile and assume $\sigma$ is not an LSE. Clearly, $\sigma$ is an NE. Then there exists a coalition $C=\{u_1,\dots,u_j\}$ of $j\le k$ players and a strategy profile $\sigma'$ which satisfy the conditions of Lemma \ref{lm:tech2}. We will show that the size of the cut increases by exactly $j$ from $\sigma$ to $\sigma'$, which is a contradiction.

First of all, observe that all the edges between players of $C$ are in the cut both in $\sigma$ and $\sigma'$. Moreover, from property (ii) of Lemma \ref{lm:tech1}, we have that the utility of each $u_i \in C$ increases exactly by one. Let $E_i=\{\{u_i,v\} | v \notin C\}$. As a consequence, we have that, for each $u_i \in C$, the number of edges in $E_i$ crossing the cut increases by exactly one. Since $E_i$ and $E_j$ are disjoint for $i \neq j$, the size of the cut increases exactly by $j$ from $\sigma$ to $\sigma'$.
\end{proof}

\begin{figure}[t]
\begin{center}
\includegraphics[scale=0.4]{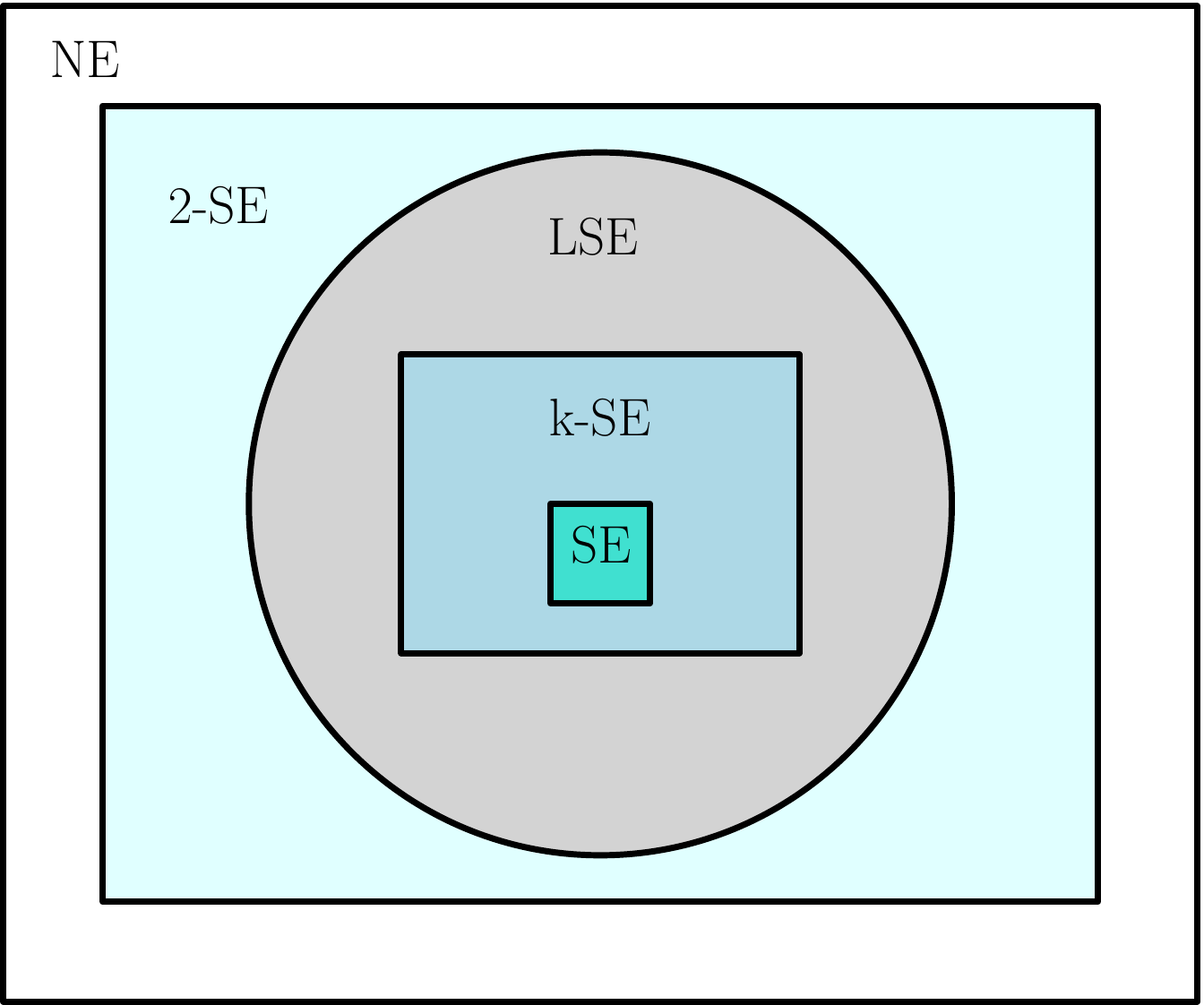}\caption{Equilibria in the unweighted max-k-cut game}\label{fig:characterization}
\end{center}
\end{figure}

We conclude this section by discussing some consequences of our analysis about how LSE is related to $q$-SE: these results are depicted in Figure \ref{fig:characterization}.
 Some inclusions are straightforward from the definition of $q$-SE. Here we show that an LSE is always a 2-SE and a $k$-SE is always an LSE. Concerning the former fact, note that a coalition of 2 players can profitably deviate from an NE $\sigma$ if and only if there exists an edge between them. In fact, otherwise, they could profitably deviate alone from $\sigma$. This means that such a coalition is a clique of two players, and hence a local coalition. As far as the latter relation is concerned, we prove the following:

\begin{prop}
A $k$-SE is always an LSE.
\end{prop}
\begin{proof}
Let $\sigma$ be a $k$-SE and, by contradiction, let $C$ be a clique which would profit deviating to $\sigma'$. By Lemma \ref{lm:tech2} there exists a minimal subcoalition $C'$ of at most $k$ players which can profit deviating alone, which is a contradiction.
\end{proof}

It is worth noticing that, as a consequence, when $k=2$ the set of LSE coincides exactly with the set of 2-SE. On the other hand, for $k \geq 3$ all inclusions are proper, as we show in \nameref{app:properinclusions}.

\section{Existence of SE for special cases}\label{sec:special cases}

In this section we show that an SE always exists for some special classes of unweighted graphs. More precisely, we prove that in graphs with large girth or large degree, any optimal strategy profile is an SE. It is worth noticing that, for general graphs, we have already proved that any optimal coloring is both a 5-SE and an LSE. We conjecture that it is indeed always an SE, even if this seems to be challenging to prove in general. A natural approach could be that of using the size of the cut as a strong potential function, that is $\Phi_S(\sigma) = S(\sigma)$, as it has already been done for proving that max k-cut games admit a Nash equilibrium \cite{DBLP:phd/de/Hoefer2007} \cite{DBLP:conf/sagt/KunPR13}.
However, it can be argued that this approach cannot work in general, since a profitable coalition deviation could sometimes result in a cut-value decrease. This is stated in the following proposition whose proof can be found in the appendix.

\begin{prop}\label{prop:phinopotential}
The size of the cut is not a strong potential function for the max $k$-cut game on unweighted graphs.
\end{prop}

Even though there exist strong improvements that can decrease $\Phi_S$, it does not mean that such function cannot be used in some interesting special setting. Indeed, there are cases in which $\Phi_S$'s value always increases after a strong improvement, that is, they admit a strong potential function. From now on we assume that only minimal coalitions can deviate.

\subsubsection{Bounded girth}

Given a graph $G$, let $\rho(G)$ be its \emph{girth}, that is the size of the minimum cycle. 
We show that a graph with girth $\rho(G)$ always admits a q-SE, for $q \leq 2\rho(G) - 3$. 
This implies that when $\rho(G) \geq \left(\lvert V\rvert + 3\right)/2$ then there always exists a strong equilibrium.

\begin{prop}\label{prop:girth}
Given an unweighted graph $G$ with girth $\rho(G)$ and any number of colors $k$, an optimal coloring is a $\left(2\rho(G) - 3\right)$-SE.
\end{prop}

\begin{corollary}\label{cor:girth}
If $\rho(G) \geq \left(\lvert V\rvert+3\right)/2$, then an optimal coloring is always an SE.
\end{corollary}

\subsubsection{Bounded degree}
Here we show that whenever the number of colors $k$ is large enough with respect to the maximum degree of the graph, then any optimal strategy profile is an SE. More precisely, we prove the following:

\begin{prop}\label{prop:degree}
Any optimal strategy profile is an SE when $k\geq \ceil*{\left(\delta^{M} + 1\right)/2}$.
\end{prop}

\begin{proof}
Let $\sigma^*$ be an optimal strategy profile and assume $\sigma^*$ is not an SE. Then a coalition $C$ and a strategy profile $\sigma'$ exist such that all players in $C$ strictly improves their utility by deviating to $\sigma'$. We will show that in this case the size of the cut will strictly increase in $\sigma'$, which contradicts the optimality of $\sigma^*$.

As we already pointed out, $\sigma^*$ is NE. Moreover, consider any node $u$. Since its degree $\delta^u$ is at most $\delta^M \le 2k -1$, we have that in any coloring, by the pigeonhole principle, there must exist a color that appears at most once in $u$'s neighborhood. As a consequence, since $\sigma^*$ is an NE, it holds that $\mu_u(\sigma^*)\ge \delta^u -1$. On the other hand, since all nodes in $C$ must strictly improve their utility, we have that, for every $u \in C$, $\mu_u(\sigma^*)=\delta^u -1$ and  $\mu_u(\sigma')=\delta^u$. This implies that the size of the cut must strictly increase. Indeed, consider the edge set $F=\{\{u,v\}| u \in C \text{ or } v \in C\}$. Clearly, only edges in $F$ can enter or leave the cut when the strategy profile changes from $\sigma^*$ to $\sigma'$. Moreover, all edges in $F$ belong to the cut $E(\sigma')$ while there is at least an edge that is not in $E(\sigma^*)$. 
\end{proof}

\section{Conclusions and future work}
We investigated coalition resilient equilibria in the max $k$-cut game. We solved an open problem proposed in \cite{DBLP:conf/tamc/GourvesM10} on weighted graphs by showing that improving deviations performed by minimal coalitions can cycle. We then provided some positive results on unweighted graphs. More precisely, we proved that any optimal coloring is both a 5-SE and a 1-LSE. We also showed that SE exist for some special cases, namely, when the graph has a large girth or the number of colors is large enough with respect to the maximum degree.

Even though we made a progress on the topic, the problem of understanding whether any instance of the max $k$-cut game admits strong equilibria is still open on both weighted and unweighted graphs. We conjecture that an optimal strategy profile is always an SE in the unweighted case. However, proving that seems to be really challenging. Another possible way to prove the existence of an SE would be that of providing a strong potential function. We proved in Proposition \ref{prop:cyclepotential} that such function cannot exist on weighted graphs even when only minimal coalitions can deviate but it is still unknown whether a strong potential function exists or not on unweighted graphs.
Along this direction, an interesting intermediate step could be that of proving the existence of $q$-SE for possibly non-constant values of $q > 5$.

Regarding $x$-local coalitions, our results are only about the case $x=1$ on unweighted graphs. Some other research questions could be the study of the existence of $x$-local strong equilibrium for $x >1$, and how to extend our results to weighted graphs. For instance, it would be interesting to investigate whether any instance of the max $k$-cut game on weighted graphs admits local strong equilibria.

\bibliographystyle{abbrv}
\bibliography{Bibliography}

\newpage
\appendix

\section*{Appendix A: omitted proofs}\label{app:omitted}

\subsection*{Proof of Proposition \ref{prop:propminimalcoalition}}

Let us start with the proof of property (i). In order to prove the equality we show that $C_K(\sigma')\subseteq C_K(\sigma)$ and $C_K(\sigma)\subseteq C_K(\sigma')$.
\begin{itemize}
\item if $v \in C$ and $\sigma'(v) \notin C_K(\sigma)$, then it is easy to see that $v$ can deviate alone and this is a contradiction with the fact that $C$ is minimal;
\item if there is a color $i \in C_K(\sigma)\backslash C_K(\sigma')$, then again it is easy to see that the coalition $C\backslash C_i(\sigma)$ can perform a strong improvement and this is a contradiction with the fact that $C$ is minimal.
\end{itemize}

In order to prove property (ii), we first show the following lemma that will subsequently be used. 
\begin{lemma}\label{lemma:acyclic}
In undirected weighted graphs any acyclic coalition $C$, with $\lvert C\rvert > 2$, that can perform a minimal strong improvement must be placed in no more than two colors, that is $\lvert C_K(\sigma)\rvert \leq 2$.
\end{lemma}

\begin{proof}
Given a minimal acyclic coalition $C$, if $\sigma'$ is the resulting coloring after $C$ deviates, we know from property (i) that $C_K(\sigma') = C_K(\sigma)$. We build a directed graph $G' = (C, E')$ where an edge $(v,u)$ is in $E'$ if and only if
(i) $\{v,u\}\in E$,
(ii) $\sigma'(v) = \sigma(u)$, and
(iii) if $u$ does not change color, then $v$ has no interest in deviating, that is, $\mu_v(\sigma') - w(\{v,u\}) \leq \mu_v(\sigma)$. In other words, the node $v$ has interest in changing color only if also node $u$ changes color. 
$G'$ represents all those moves that are mandatory in order to make the deviation profitable to every player in $C$. Since $C$ performs a minimal strong improvement, every player $v \in C$ must have at least one ingoing edge and one outgoing edge, since otherwise $v$ could be removed from $C$, i.e., there exists a sub coalition of $C$ not containing $v$ that can perform a strong improvement, or $v$ could deviate alone, respectively. This implies the existence of at least one oriented cycle in $G'$, and consequently the existence of a cycle in $G(C)$ too. Since we know that $G(C)$ is acyclic, the only type of cycle that can happen in $G'$ is the one limited to adjacent vertices, and this happens only when the number of colors used by $C$ is at most 2.
\end{proof}

Lemma \ref{lemma:acyclic} says that a minimal coalition $C$ must use at most two colors when $G(C)$ is acyclic. If $\lvert C_K(\sigma)\rvert = 1$ then $\lvert C\rvert = 1$ by the minimality condition, and we already know that the cut-value always strictly increases after a Nash improvement \cite{DBLP:conf/stacs/ChristodoulouMS06}. If $\lvert C_K(\sigma)\rvert = 2$ we know from \cite{DBLP:conf/wine/GourvesM09} that each strong improvement performed by a coalition that uses only two colors (i.e., the classical max cut problem) always increases the size of the cut.

\qed

\subsection*{Proof of Proposition \ref{prop:propcoloring}}
By hypothesis we know that $|C_K(\sigma)| \in \{2, \lvert C\rvert - 1, \lvert C\rvert \}$.
Let us consider the three cases in detail: 

\begin{itemize}
\item $|C_K(\sigma)| = 2$. This result derives from \cite{DBLP:conf/wine/GourvesM09}, which shows that in the max cut game (i.e., $k$ = 2), any strong improvement always strictly increases the size of the cut. 	
	\item $|C_K(\sigma)| = \lvert C\rvert$. For each color $i\in C_K(\sigma)$ there is a unique player $u \in C$ such that $\sigma(u) = i$. By Proposition \ref{prop:propminimalcoalition}, property (i), we know that no pair of players can share the same color in $\sigma'$, since otherwise $C$ would not be minimal. Moreover, for each player $u\in C$ it holds by Lemma \ref{lemma:2} that $\delta_u^{\sigma(u)}(\sigma) = \delta_u^{\sigma'(u)}(\sigma)$, namely the number of edges that leave the cut is equal to the number of edges towards non-deviating players that enter the cut.
In addition, since players' utilities must increase there must be at least $\lvert C\rvert -1$ more edges in the cut than the one in $\sigma$, that is the cut-value increases.
	\item $|C_K(\sigma)| = \lvert C\rvert - 1$. There must be exactly two vertices colored the same in $\sigma$, let them be $u$ and $v$. Analogously, in $\sigma'$ there must be exactly two nodes colored the same too, let them be $x$ and $y$. There are three possible cases:
	\begin{itemize}
		\item  $u = x$ and $v = y$, that is $u$ and $v$ stay together in both $\sigma$ and $\sigma'$. If $\{u,v\}\notin E$, then by Lemma \ref{lemma:2} $|C_K(\sigma)| = \lvert C\rvert$, $\delta_u^{\sigma(u)}(\sigma) = \delta_u^{\sigma'(u)}(\sigma)$ and $\delta_v^{\sigma(v)}(\sigma) = \delta_v^{\sigma'(v)}(\sigma)$ and similarly to the case $|C_K(\sigma)| = \lvert C\rvert$, $u$ and $v$ provide a new edge each to the cut which increases its size. Conversely, if $\{u,v\}\in E$ then again by Lemma \ref{lemma:2} at least one player among $u$ and $v$ does not improve its utility after the deviation, thus $C$ is not a strong improvement.
		\item $u \neq x$ and $v \neq y$, that is in $\sigma'$ two new nodes $x,y\in C$ share the same color, while $u$ and $v$ are now separated. If $\{x,y\}\in E$, then this edge leaves the cut in $\sigma'$. If $x$ and $y$ switch to a color from which only one player deviates, then by Lemma \ref{lemma:2} their utility do not strictly increase. Thus, it must necessarily be that $x$ and $y$ benefit from the deviation of at least two players, and the only chance is $u$ and $v$. Therefore, $\sigma'(x) = \sigma'(y)  = \sigma(u) = \sigma(v)$, and $\{x,u\}, \{x,v\}, \{y,u\}, \{y,v\} \in E$. In order to strictly increase $x$ and $y$'s utility it must be that $\delta_x^{\sigma(x)} = \delta_x^{\sigma(u)}$ and $\delta_y^{\sigma(x)} = \delta_y^{\sigma(u)}$, because they earn two edges and they lose edge $\{x,y\}$. Moreover, we know from Lemma \ref{lemma:2} that $\delta_u^{\sigma(u)} = \delta_u^{\sigma'(u)}$ and
$\delta_v^{\sigma(u)} = \delta_v^{\sigma'(v)}$. So, four new edges enter the cut and edge $\{x,y\}$ leaves the cut, that is the cut-value increases. Conversely, if $\{x,y\}\notin E$, then the potential increases in a way similarly to the one described in the previous case.
		\item One node among $u$ and $v$, say $u$, shares color $\sigma'(u)$ with another node $x\in C$. It must be that $\{u,x\}\notin E$ otherwise, by Lemma \ref{lemma:2}, $x$ would have no desire to deviate. Thus, the cut-value must increase.
	\end{itemize}
\end{itemize}
\qed

\subsection*{Proof of Theorem \ref{prop:propcoalition}}
We will argue that any improving deviation performed by a coalition of size at most 5 strictly increases the size of the cut. 

Let $C=\{u,v,w,x,y\}$ be a minimal coalition of size five that can deviate from a stable coloring $\sigma$ inducing coloring $\sigma'$. We already know from Proposition \ref{prop:propcoloring} that when $C$ is placed in either two, four or five colors in $\sigma$, then any of its minimal improvement always increases the cut-value. Thus, the remaining case is when $\lvert C_K(\sigma)\rvert = 3$. Let the colors used be $1$, $2$, and $3$. Let us consider how the nodes in $C$ are located in $\sigma$.

First, there could be three nodes $u,v,w$ sharing the same color $1$, while nodes $x$ and $y$ are colored $2$ and $3$, respectively. Since $u,v$ and $w$ have to choose which color to deviate among $2$ and $3$, at least two of them have to be colored the same in $\sigma'$ too. Let us assume without loss of generality that $u$ an $v$ share the same color in $\sigma'$. If $\{u,v\}\in E$, then by Lemma \ref{lemma:2} they do not improve their utility because they benefit from the deviation of only one player. This implies that the nodes in $\{u,v,w\}$ that are colored the same in $\sigma'$ cannot have any edge between them.
Moreover, if one of the three nodes, say $w$, deviates to the same color as another player of the coalition, say $x$, then $\{w,x\}\notin E$ again by Lemma \ref{lemma:2}. Thus, the potential value increases for such configuration.

The other case is when two nodes $u,v$ are colored $1$, two nodes $w,x$ are colored $2$ and the fifth node $y$ is colored 3. Let us analyze how the players in $C$ can be located in $\sigma'$ after the deviation.
\begin{itemize}
\item There are three nodes paired together and two nodes alone. The triplet must necessarily be composed by a pair of nodes that are paired together also in $\sigma$, let them be $u,v$, plus a third node. The number of edges in the subgraph induced by such triplet is at most one. In fact, there is a player in such subgraph with degree 2 then, since it benefits from the deviation of at most two players and $\sigma$ is stable, its utility does not increase after the deviation. Moreover, there are two players in $C$ that are alone in $\sigma'$, that is they add at least one edge each to the cut. Therefore, at most one edge leaves the cut and at least two edges enters the cut, namely $\Phi$ increases.
\item There are two pairs of nodes together, and a node alone. If the node alone is $y$ in both $\sigma$ and $\sigma'$,
then the two pairs of nodes in $\sigma'$ must be the same in $\sigma$, that is $\sigma'(u)=\sigma'(v)$ and $\sigma'(w)=\sigma'(x)$. One of the two pairs, say $u,v$ deviates to $\sigma(y)$, that is it benefits only from the deviation of $y$, and by Lemma \ref{lemma:2} this implies that $\{u,v\}\notin E$. The other pair $w,x$ deviates to $\sigma(u)$, and if $\{w,x\}\in E$ it implies a deviation similar to the one described in Proposition \ref{prop:propcoloring} for $|C_K(\sigma)| = \lvert C\rvert - 1$, case 2, and we know that such type of deviation increases the cut-value. If $\{w,x\}\notin E$, then the cut-value increases similarly to Proposition \ref{prop:propcoloring}, case 1 of $|C_K(\sigma)| = \lvert C\rvert -1$.  Conversely, suppose that $y$ is now paired with another node, say $u$. If $\{y,u\}\in E$ then in order to make player $y$ deviate profitably, it must have two edges $\{y,w\}, \{y,x\}$ towards color $\sigma(w)$, and it must be that $\sigma'(y) = \sigma(w)$. As described in previous cases, this kind of deviation increases the cut-value. If $\{y,u\}\notin E$, then by Lemma \ref{lemma:2} $y$ and $u$ add at least two edges to the cut. The other pair is necessarily colored $\sigma(y)$, that is the two players benefit only from $y$'s deviation. Lemma \ref{lemma:2} implies that they cannot have an edge among them, therefore the edges in the cut increases.
\end{itemize}
\qed

\subsection*{Proof of Lemma \ref{lm:tech1}}
Let $u \in C$ be a node in the coalition and let $\nu=\sigma'(u)$ be the new color of $u$. 
We claim that
\begin{equation}\label{eq:eq1}
0< c_u(\sigma)-c_u(\sigma')\leq  c_{u,C}(\sigma_{-u}, \nu)- c_{u,C}(\sigma')
\end{equation}
\noindent Indeed, first notice that $u$ strictly decreases its cost deviating to $\sigma'$, and hence $c_u(\sigma') < c_u(\sigma)$. Moreover, since $\sigma$ is an NE we also have that  $c_u(\sigma) \leq c_u(\sigma_{-u}, \nu)$. As a direct consequence we get the following double inequality:
\[
0<c_u(\sigma)-c_u(\sigma')\leq c_u(\sigma_{-u}, \nu)-c_u(\sigma').
\]
Since the players outside the coalition do not change their colors, it holds that $(c_u(\sigma')-c_{u,C}(\sigma'))=(c_u(\sigma_{-u}, \nu)-c_{u,C}(\sigma_{-u}, \nu))$. And hence we have
$c_u(\sigma')=c_u(\sigma_{-u}, \nu)-c_{u,C}(\sigma_{-u}, \nu)+c_{u,C}(\sigma')$, which implies:
 \begin{align}
 & c_u(\sigma_{-u}, \nu)-c_u(\sigma') \; = \notag\\
 = \;& c_u(\sigma_{-u}, \nu)-(c_u(\sigma_{-u}, \nu)-c_{u,C}(\sigma_{-u}, \nu)+c_{u,C}(\sigma')) \; = \notag\\
 = \;& c_{u,C}(\sigma_{-u}, \nu)-c_{u,C}(\sigma'). \notag
 \end{align}
Replacing this expression in the double inequality we have derived before, we get Equation \eqref{eq:eq1}.

Since $C$ is a clique and $\sigma'(u)=\nu$ we know also that $c_{u,C}(\sigma')=|C_{\nu}(\sigma')|-1$ while $c_{u,C}(\sigma_{-u}, \nu)=|C_{\nu}(\sigma)|$. Replacing these terms in Equation \eqref{eq:eq1} we obtain:
\begin{equation}\label{eq:eq2}
0< c_u(\sigma)-c_u(\sigma')\leq  |C_{\nu}(\sigma)|-|C_{\nu}(\sigma')|+1
\end{equation}
which implies that $|C_{\nu}(\sigma')| \leq |C_{\nu}(\sigma)|$. The arguments used so far hold for any player $u \in C$, and thus the above inequality also holds that for all colors $i \in C_K(\sigma')$, i.e. $|C_{i}(\sigma')| \leq |C_{i}(\sigma)|$ for each $i \in C_K(\sigma')$. Observe that this implies also that 
$C_K(\sigma') \subseteq C_K(\sigma)$.

By summing up over all colors in $C_K(\sigma)$, we have that:
\[
\sum_{i \in C_K(\sigma)} |C_i(\sigma)| = \sum_{i \in C_K(\sigma)} |C_i(\sigma')| = |C|
\]
and as a consequence we get $|C_i(\sigma)|=|C_i(\sigma')|$ for all $i \in C_K(\sigma)$ that is exactly (i).
Thus we can rewrite Equation \eqref{eq:eq2} as:
\[
0< c_u(\sigma)-c_u(\sigma')\leq 1
\]
which implies (ii).

We now prove (iii).  
We note that:
\begin{align*}
c_u(\sigma') & =  |\{(u,v)| v \notin C, \sigma'(v) = \nu \}| +  |\{(u,v)| v \in C \setminus\{u\}, \sigma'(v) = \nu\}| \\
& =  |\{(u,v)| v \notin C, \sigma(v) = \nu \}| + |C_{\nu}(\sigma')| - 1 \\
& =  |\{(u,v)| v \notin C, \sigma(v) = \nu \}| + |C_{\nu}(\sigma)| - 1 \\
& = c_u(\sigma_{-u}, \nu) - 1
\end{align*}
\noindent where in the last but on equality we used (i). 
Since by (ii) we know that $c_u(\sigma) = c_u(\sigma') + 1$, we can conclude (iii).
\qed

\subsection*{Proof of Proposition \ref{prop:phinopotential}}

\begin{figure}[h]
\begin{center}
\includegraphics[scale=0.4]{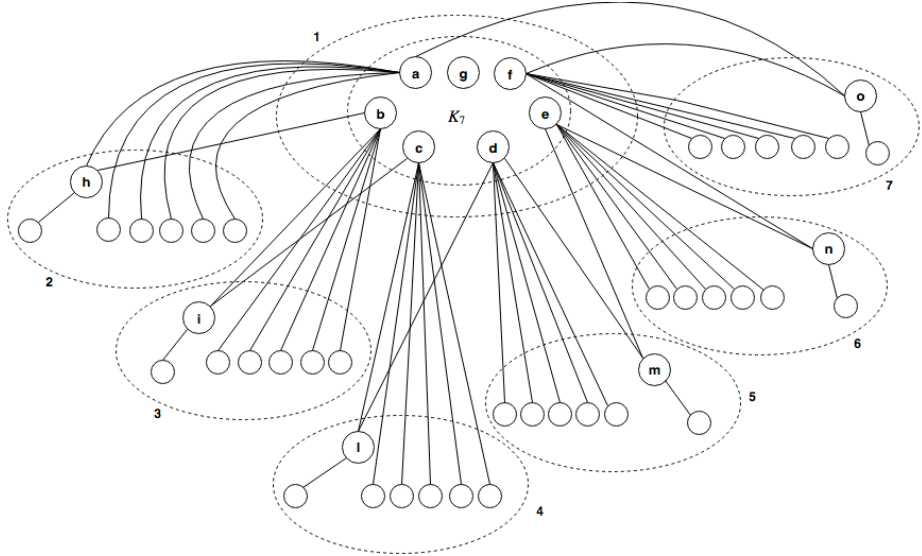}
\caption{Instance $G$ and coloring $\sigma$ for which the minimal strong improvement strictly decreases the cut-value.}\label{fig:unweightedlowerpotential}
\end{center}
\end{figure}

Consider the graph $G$ and a coloring $\sigma$. Both are depicted in Figure \ref{fig:unweightedlowerpotential}. $K_7$ denotes a clique composed by nodes $a, b, \ldots, g$. Due to space shortage, only the needed edges are explicitly reported. Regarding the omitted edges, we assume that 
(i) none of the non-labelled nodes vertices desire to change color,
(ii) node $g$ does not want to deviate to any other color,
(iii) nodes $h, i, l, m, n,$ and $o$ have sufficient edges towards all the other colors, except color $1$, to prevent them from deviating to such colors, and
(iiii) $a, b, c, d, e$, and $f$ have sufficient edges towards all the other colors, except colors $2, 3, 4, 5, 6$, and $7$ respectively, to prevent them from deviating to such colors.
Given these assumptions, the only minimal coalition that can perform an improving move is $C=\{a,b,c,d,e,f,h,i,l,m,n,o\}$. In fact, $a$ has six edges towards both color $1$ and color $2$, thus it is willing to deviate to color $2$ only one of its neighbors leaves that color. Node $h$ has one edge towards color $2$ and two edges towards nodes $a$ and $b$, that are colored $1$. Thus, $h$ moves to color $1$ only if players $a$ and $b$ deviate. Similarly to node $a$, $b$ deviates to color $3$ only if $i$ changes color, and so on and so forth. Finally, player $o$ deviates to color $1$ only if $a$ and $f$ change color. To sum up, the coloring $\sigma'$ obtained after $C$'s deviation is such that $\sigma'(a) = 2, \sigma'(b) = 3, \sigma'(c) = 4, \sigma'(d) = 5, \sigma'(e) = 6, \sigma'(f) = 7, \sigma'(h) = \ldots = \sigma'(o) = 1$, while the other nodes' strategies remain unchanged. It is not difficult to check that $\Phi(\sigma') - \Phi(\sigma) =-3$ that is, $\Phi$'s value strictly decreases even after a strong minimal improvement. 
\qed

\subsection*{Proof of Proposition \ref{prop:girth}}

\begin{figure}[h!]
\begin{center}
\includegraphics[scale=0.35]{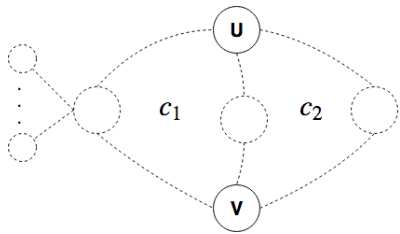}
\caption{A coalition $C$ whose size is upper-bounded as defined in Proposition \ref{prop:girth} has at most two cycles.}\label{fig:girth}
\end{center}
\end{figure}

Let $\sigma^*$ be an optimal coloring and let $C$ be a minimal coalition that deviates from $\sigma^*$ to $\sigma'$, where $\lvert C\rvert \leq \left(2\rho(G) - 3\right)$. Let us consider how many cycles there can be at most in $G(C)$ according to $\lvert C\rvert$. 
\begin{itemize}
\item Trivially, if $\lvert C\rvert < \rho(G)$ then $G(C)$ must not contain any cycle;
\item If $\lvert C\rvert \geq \rho(G)$ then $G(C)$ can contain at least one cycle;
\item Given a cycle, since it has length at least $\rho(G)$ there always exist a pair of nodes $u,v$, such that $dist_{G(C)}(u,v) \geq \floor*{\rho(G)/2}$, that is the shortest path between $u$ and $v$ is made of at least $\floor*{\rho(G)/2} + 1$ nodes, $u$ and $v$ included. Thus, by having at least $\ceil*{\rho(G)/2} - 1$ new nodes then $G(C)$ is allowed to have a second cycle that makes use of the shortest path between $u$ and $v$;
\item Since $\lvert C\rvert \leq \left(2\rho(G) - 3\right)$, any coalition $C$ cannot contain more than two cycles in $G(C)$, otherwise there would be a cycle of length strictly less than $\rho(G)$.
\end{itemize}
Let us assume that $G(C)$ contains two cycles $c_1$, $c_2$ with some nodes in common. Let $u$ and $v$ be the first and last node in the intersection of the cycles, respectively. Moreover, since there cannot be any more cycles, each node of the cycles can be the root of a (possible empty) tree of deviating players (if we do not consider the other nodes in the cycles). A sketch of $G(C)$ is depicted in Figure \ref{fig:girth}.

First, we know from Proposition \ref{prop:propminimalcoalition} that the deviation performed by a player $w$ that belongs to $C$ but not to $c_1$ and $c_2$ always increases the size of the cut. If it was not so then it would mean that $w$ is not increasing its utility because, since it does not belong to any of the cycles, its adjacent players who deviate always select a color different from its own in $\sigma'$, otherwise some node could be removed from $C$, contradicting the minimality assumption.

Let $w \in c_1\cup c_2$ be a node that belongs to one of the two cycles. The only $w$'s deviating neighbors that can be in the same color as $w$ in $\sigma'$ are the ones that belong to $c_1$ or $c_2$. In fact, if there is a node $x$ that is adjacent to $w$, is not in $c_1$ or $c_2$ and $\sigma'(x) = \sigma'(w)$, then by removing from $C$ the subtree with root $x$ we still have a deviating coalition. For the same reason, there cannot be a deviating node $x$ that is adjacent to $w$, is not in $c_1$ and $c_2$ and such that $\sigma^*(w) = \sigma^*(x)$, because $w$ does not need $x$ to deviate. Therefore, we can assume in the following that $C$ is composed only of $c_1$ and $c_2$.

Node $w$ must have at least one adjacent node $x$ colored differently from it in $\sigma^*$, that is $\sigma^*(w) \neq \sigma^*(x)$, otherwise $w$ would have no reason to deviate jointly because it could perform a Nash deviation itself. Moreover, since (i) $2\leq \delta^w(G(C)) \leq 3$ $\forall w\in C$, (ii) each deviating player has at least one proper edge in $G(C)$ with respect to $\sigma$, (iii) by Lemma \ref{lemma:2}, players with degree $2$ cannot share the same color with some other deviating nodes in $\sigma'$, otherwise they could have deviated alone, and (iiii) there are only two not-neighboring players $u, v$ with degree 3 in $G(C)$, it must necessarily be that no edge in $G(C)$ is monochromatic with respect to $\sigma'$.
Moreover, the difference between the number of properly colored edges incident to $w\in C$, where $\delta^w(G(C)) = 2$, and towards nodes outside $C$ in $\sigma'$ and $\sigma^*$, respectively, is greater or equal than 0. This is true otherwise $w$'s utility has not strictly improved after the deviation, since $w$ can improve its utility by at most $1$ from the deviation of its neighbors because $\sigma^*$ is Nash stable by definition. Thus, the deviating players with degree in $C$ equal to 2 contributes to the cut a total of $\lvert C\rvert -2$ edges, since they all improve their utility by 1 each. Since we know from Theorem \ref{prop:propcoalition} that $\sigma^*$ is a 5-SE, we can assume $\lvert C\rvert > 5$, that is more than three edges enters the cut.
On the other hand, the only two nodes $u,v$ with degree $\delta^u(G(C)) = \delta^v(G(C)) = 3$ can decrease the number of the edges in the cut and towards nodes not in the coalition by at most 1 each. This is true otherwise if one of them decreases the cut-value by at least 2 then it does not strictly improve its utility. 
Thus, the cut-value increases, but this contradicts the fact that $\sigma^*$ is optimal.

\qed

\section*{Appendix B}\label{app:properinclusions}

In this section we show that the inclusions depicted in Figure \ref{fig:characterization} are all proper.

\begin{figure}[ht]
\centering
\includegraphics[scale=0.2]{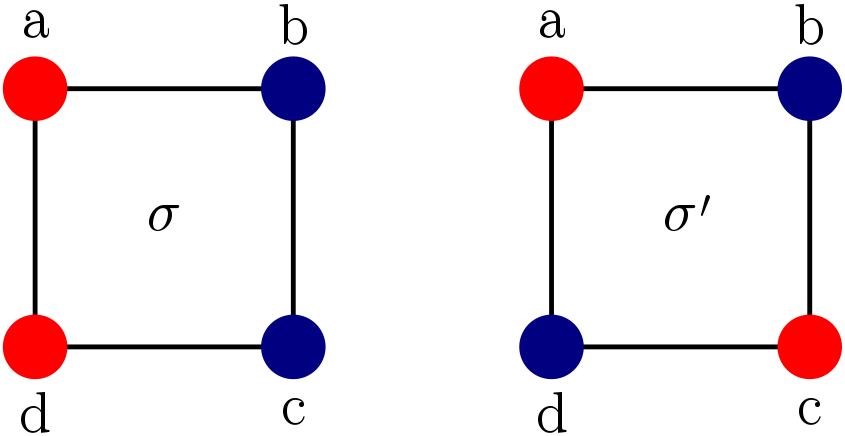}
\caption{An example of an NE which is not a 2-SE in the Max-Cut game.}
\label{fig:NE_not_2SE}
\end{figure}

First we give an example of NE which is not a 2-SE in the Max-Cut game. The example is shown in Figure \ref{fig:NE_not_2SE}. The strategy profile $\sigma$ is an NE, but the couple \{c, d\} can improve their utility by deviating to $\sigma'$.

\begin{figure}[ht]
\centering
\subfloat[Initial strategy profile.]{
\includegraphics[width=0.3\linewidth]{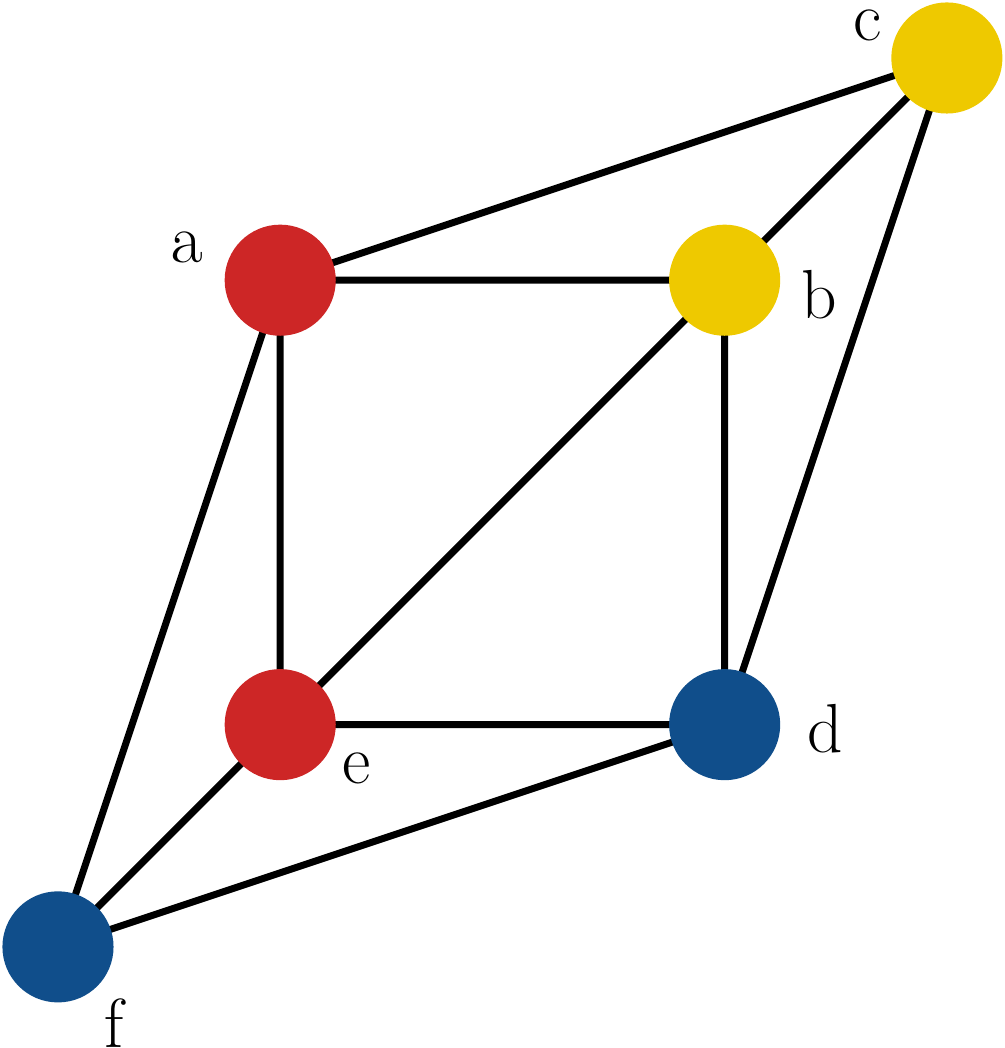}
\label{fig:2SE_notLSE1}}
\quad
\subfloat[Strategy profile after the deviation of b, d and e.]{
\includegraphics[width=0.3\linewidth]{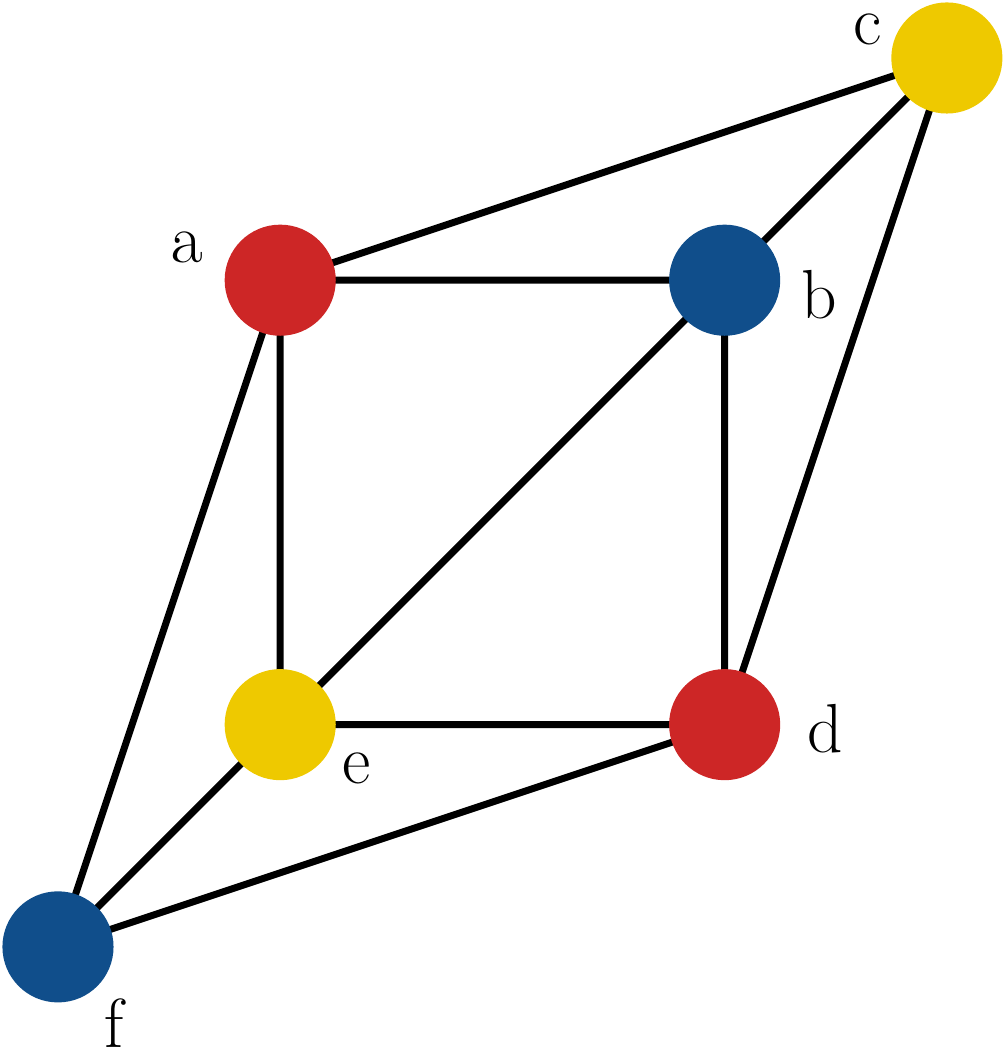}
\label{fig:2SE_notLSE2}}
\caption{Example of a 2-SE which is not an LSE in the Max-3-Cut game.}
\label{fig:2SE_notLSE}
\end{figure}

In Figure \ref{fig:2SE_notLSE} we have an example of a 2-SE (Figure \ref{fig:2SE_notLSE1}) which is not an LSE. Figure \ref{fig:2SE_notLSE2} shows a profitable deviation of players b, d and e, which form a clique of three nodes. Thus the initial strategy profile is not a LSE.

\begin{figure}
\centering
\subfloat[Initial strategy profile.]{%
\includegraphics[width=0.35\linewidth]{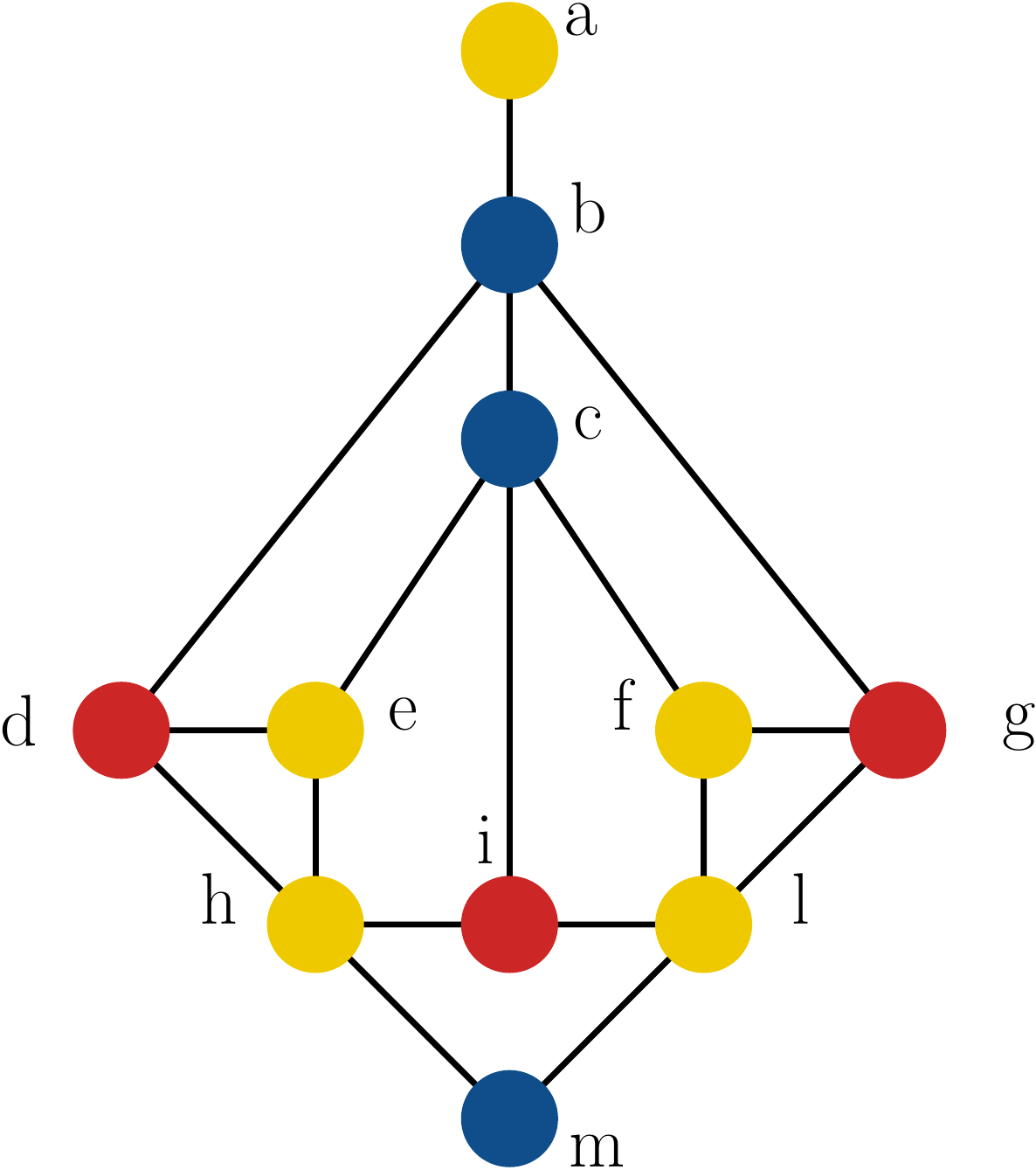}
\label{fig:LSE_not_3SE1}}
\quad
\subfloat[Strategy profile after the deviation of c, e and f.]{%
\includegraphics[width=0.35\linewidth]{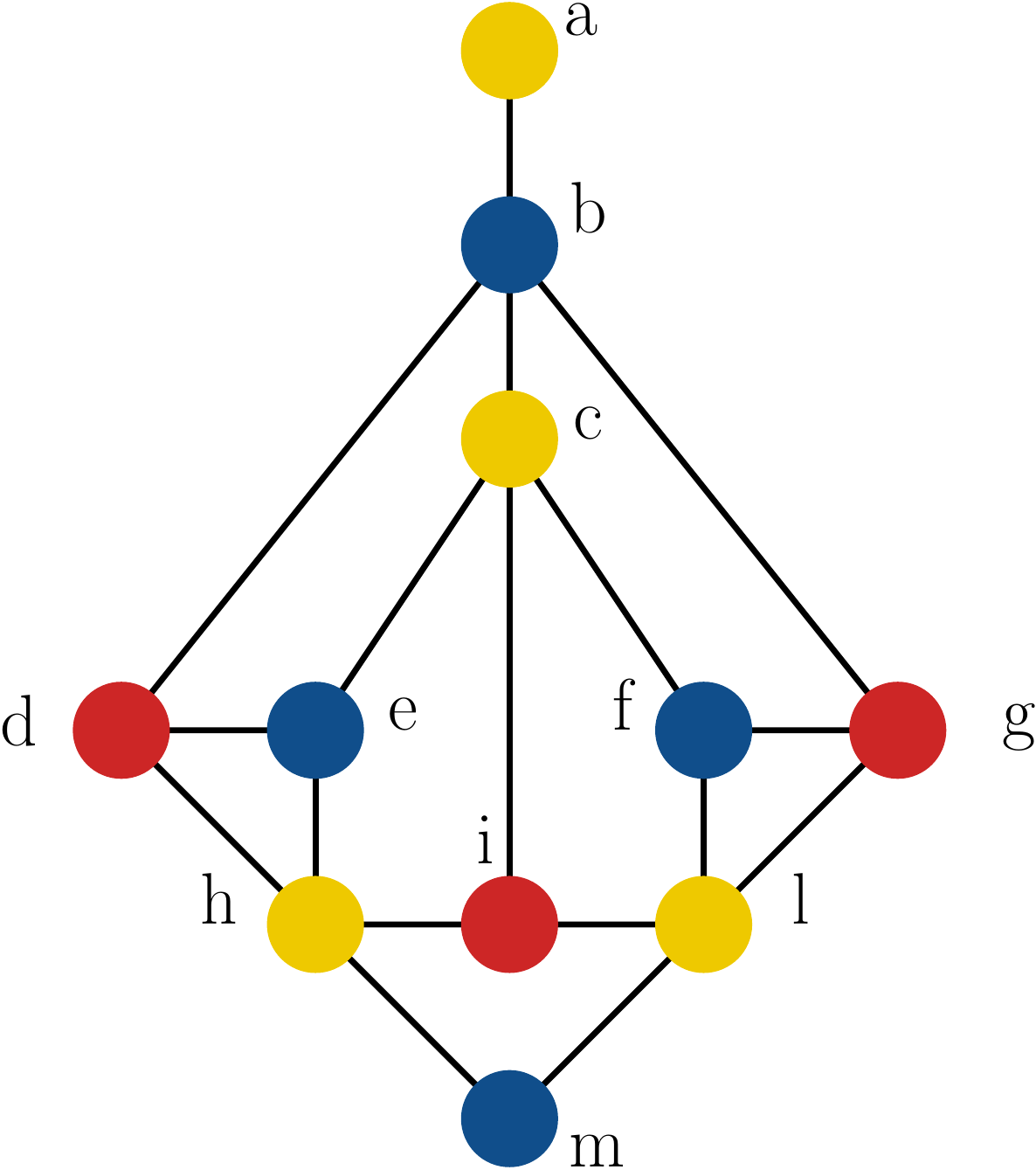}
\label{fig:LSE_not_3SE2}}
\caption{Example of a LSE which is not a 3-SE in the Max-3-Cut game.}
\label{fig:LSE_not_3SE}
\end{figure}

Now we show an example of a LSE being not a $k$-SE. In the following example $k=3$. The example is given in Figure \ref{fig:LSE_not_3SE}.
 It is easy to check that the strategy profile shown in Figure \ref{fig:LSE_not_3SE1} is an NE. Moreover, the nodes \emph{a, d, g, i, m} have the highest possible utility, thus no deviating coalition can contain them, included the two cliques of 3 nodes. All possible couples either contain one of the players mentioned above, or does not satisfy the conditions given in \ref{lm:tech1}, thus they do not profit deviating. Thus, the strategy profile is an LSE. On the other hand, if \emph{e} and \emph{f} change to \emph{blue} and \emph{c} to \emph{yellow} they all improve by 1 their payoff. Thus \{\emph{c, e, f}\} is a non-local coalition which breaks the equilibrium.

\begin{figure}[b]
\centering
\includegraphics[width=0.5\linewidth]{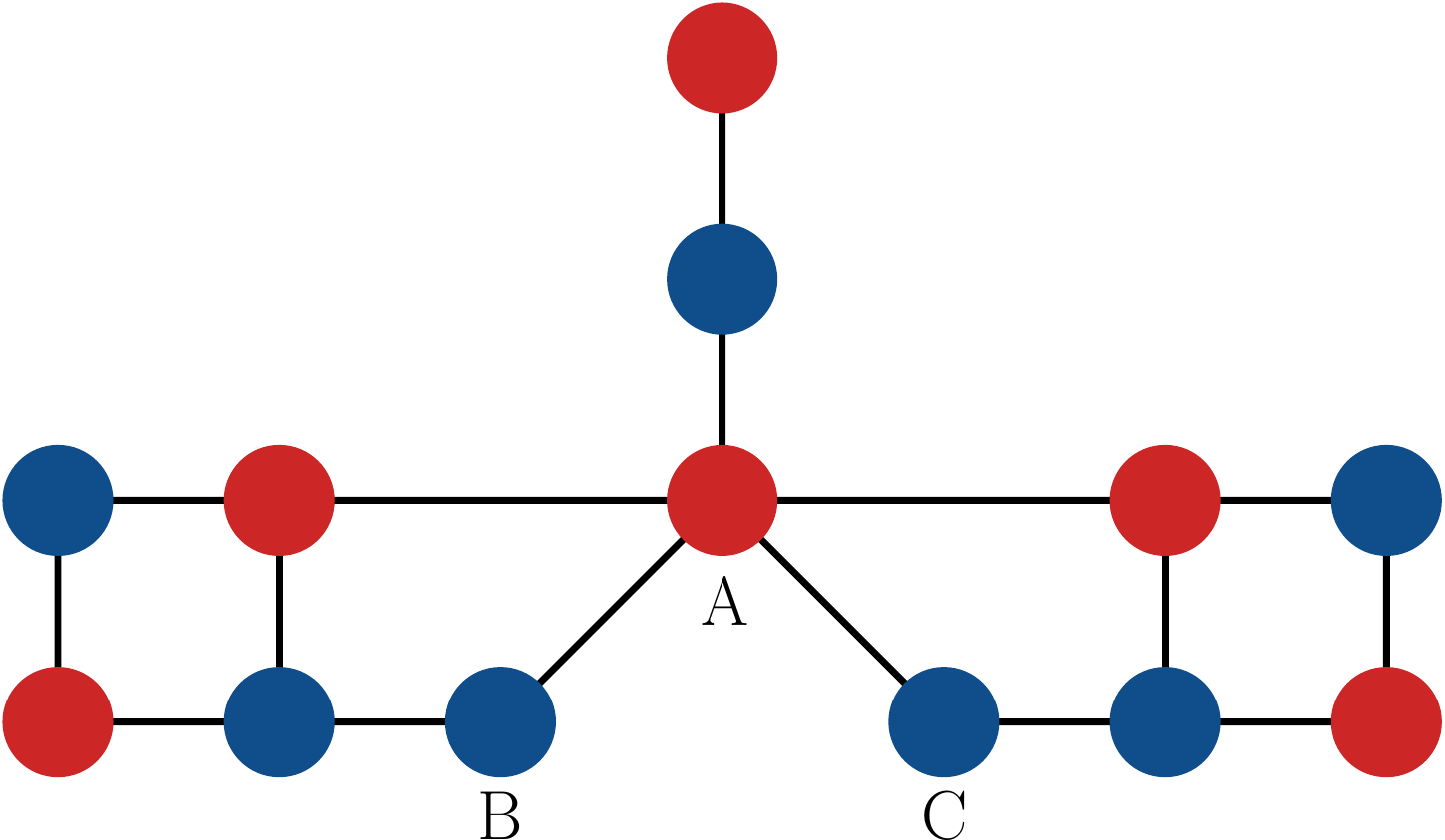}
\caption{Example of a 2-SE which is not an SE in the Max-Cut game.}
\label{fig:MaxCut_counter}
\end{figure}

There is only one case left. In Figure \ref{fig:MaxCut_counter} is given an instance of the Max-Cut game along with a 2-SE which is not an SE (if A, B and C deviate in the only possible way they all improve their utility by 1).

\end{document}